\documentclass{IOS-Book-Article}

\usepackage{mathptmx}
\usepackage{soul}\setuldepth{article} 
\usepackage[utf8]{inputenc}
\usepackage{amsmath,amsthm}
\usepackage{graphicx}
\usepackage{booktabs}
\usepackage{multirow}
\usepackage{url}
\usepackage{hyperref}
\usepackage{xurl}
\usepackage[table]{xcolor}
\usepackage{enumitem}

\usepackage{tabularx}
\usepackage{amssymb}
\usepackage{stmaryrd}  
\usepackage{tikz}
\usetikzlibrary{
  arrows.meta,
  positioning,
  fit,
  backgrounds,
  shapes.geometric,
  shapes.multipart
}
\usepackage{listings}
\newcommand{\activates}{\mathit{activates}}
\newcommand{\founds}{\mathit{founds}}
\newcommand{\Perm}{\mathsf{P}}         

\newcommand{\hasrem}{\mathit{has\_rem}}
\newcommand{\PermPos}{\mathsf{Permission}}   
\newcommand{\Duty}{\mathsf{Duty}}            
\newcommand{\Right}{\mathsf{Right}}     
\newcommand{\NoRight}{\mathsf{NoRight}} 
\newcommand{\Power}{\mathsf{Power}}     

\newcommand{\Disability}{\mathsf{Disability}} 
\newcommand{\Immunity}{\mathsf{Immunity}}
\newcommand{\Subj}{\mathsf{Subj}}       
\newcommand{\Proh}{\mathsf{Proh}}    

\newcommand{\Forbearance}{\mathsf{Forbearance}}
\newcommand{\odrl}[1]{\texttt{odrl:#1}}

\theoremstyle{definition}
\newtheorem{definition}{Definition}[section]
\newtheorem{axiom}{Axiom}[section]
\newtheorem{theorem}{Theorem}[section]
\newtheorem{proposition}{Proposition}[section]
\newtheorem{example}{Example}[section]

\newtheorem{corollary}{Corollary}[section]
\newtheorem{principle}{Principle}
\def\hb{\hbox to 11.5 cm{}}
\begin{document}

\pagestyle{headings}
\def\thepage{}
\begin{frontmatter}           
\title{What Does ODRL Mean? A Cross-Level Ontological Grounding of Permissions, Prohibitions, and Duties 
in UFO-L}
\markboth{}{D.M. Mustafa et al.\hb}
\author[A,D]{\fnms{Daham M.} \snm{Mustafa}%
\thanks{Corresponding Author. Emails:
\texttt{daham.mohammed.mustafa@fit.fraunhofer.de};
\texttt{christoph.lange-bever@fit.fraunhofer.de}}},
\author[A,D]{\fnms{Christoph} \snm{Lange}},
\author[B]{\fnms{Giancarlo} \snm{Guizzardi}},
\author[A,C]{\fnms{Diego} \snm{Collarana}},
\author[A,D]{\fnms{Christoph} \snm{Quix}}
and
\author[A,D]{\fnms{Stefan} \snm{Decker}}
\runningauthor{D.M. Mustafa et al.}
\address[A]{Fraunhofer FIT, Sankt Augustin, Germany}
\address[B]{University of Twente, The Netherlands}
\address[C]{Universidad Privada Boliviana, Bolivia}
\address[D]{RWTH Aachen University, Aachen, Germany}
\begin{abstract}
ODRL policy evaluators produce verdicts, but say nothing about the normative positions a policy brings into existence, the authority structures those positions presuppose, or who holds the power to declare a norm violated.
We formulate the Cross-Level Design Principle: any normative language
with violable, consequential norms requires both conduct-level
positions (Permission, Duty, Right, No right) and competence-level
positions (Power, Subjection, Immunity, Disability).
Applying this to ODRL, we establish that prohibition is
\emph{sanctioned} (violation possible and consequential), that
permission is underspecified across its \texttt{behaviour}
parameter (open vs.\ closed world), and that the formal semantics
covers achievement obligations only.
We ground ODRL in UFO-L, mapping each activated rule to a simple
legal relator and extending coverage from two to eight legal
positions; violation-declaration authority, implicit in every
existing evaluator, becomes an explicit Power, Subjection pair.
All axioms are mechanically verified in
Isabelle/HOL and across a 39-problem benchmark
under Vampire, E, and Z3.
\end{abstract}

\end{frontmatter}
\vspace{-15pt}
\section{Introduction}
\label{sec:introduction}
The Open Digital Rights Language (ODRL)~\cite{iannella2018odrl} is
a policy language standardised by W3C.\footnote{Benchmark available at
\url{https://github.com/Daham-Mustaf/odrl-ufol-grounding}} It has been adopted by the European Data Spaces~\cite{otto2019international,theissen2023semantics,gaiax-policy-reasoning-2023} as the lingua franca for expressing data usage and data access policies over shared assets. ODRL provides both an information model and an
ontology~\cite{iannella2018odrl} specifying how policies, rules, and
constraints are structured. The Formal Semantics specifies the behaviour of an \emph{ODRL Evaluator}, a software
component that determines, given a set of policies and a state of
the world, which rules are active, which actions are permitted or
prohibited, and which obligations have been fulfilled or
violated~\cite{odrl-formal-semantics,salas2025odrl}.
The evaluator's output is a verdict: \texttt{active},
\texttt{violated}, \texttt{fulfilled}, or \texttt{not-set}~\cite{odrl-formal-semantics}.

\noindent Consider \texttt{:KulturPortal}, an events portal that obtains the same theater
showtime records for a cultural event from two Data Space providers, the
theaters \texttt{:BerlinerEnsemble} and \texttt{:DeutschesTheater}, via
separate connectors and contracts. Each issues an \odrl{Agreement} prohibiting
redistribution of its showtimes, but with a different remedy:
\texttt{:BerlinerEnsemble} demands compensation, \texttt{:DeutschesTheater}
demands deletion of the copy.
\begin{center}
\begin{minipage}[t]{0.47\columnwidth}
\scriptsize
\begin{verbatim}
:BE-Agreement a odrl:Agreement ;
 odrl:prohibition [
  odrl:assignee :KulturPortal ;
  odrl:assigner :BerlinerEnsemble ;
  odrl:action   odrl:distribute ;
  odrl:target   :BE-Showtimes ;
  odrl:remedy [      # if violated:
   odrl:action
    odrl:compensate ;# pay theater
   odrl:assignee
    :BerlinerEnsemble ] ] .
\end{verbatim}
\end{minipage}%
\vrule\hspace{4pt}%
\begin{minipage}[t]{0.47\columnwidth}
\scriptsize
\begin{verbatim}
:DT-Agreement a odrl:Agreement ;
 odrl:prohibition [
  odrl:assignee :KulturPortal ;
  odrl:assigner :DeutschesTheater ;
  odrl:action   odrl:distribute ;
  odrl:target   :DT-Showtimes ;
  odrl:remedy [      # if violated:
   odrl:action
    odrl:delete ;    # destroy copy
   odrl:assignee
    :KulturPortal ] ] .
\end{verbatim}
\end{minipage}
\end{center}
\texttt{:KulturPortal} distributes both datasets. The ODRL Evaluator reports \texttt{violated} for both
policies, but cannot determine by whose normative
authority either remedy is enforceable: it sees the remedy
assignees but has no account of which provider holds the
Power to declare a violation and demand remedy performance. Neither prohibition carries a constraint, so both are
unconditionally active; since the regulated action was
performed, both reach
\texttt{violated}~\cite{odrl-formal-semantics}.
Each prohibition specifies a remedy, but the Formal
Semantics defines no evaluation procedure for remedy
duties and the Evaluation Report remains
unspecified~\cite{odrl-formal-semantics}.
The evaluator goes no further: it cannot determine
(\emph{i})~by what normative standing either provider
may demand remedy performance, or
(\emph{ii})~what normative standing grounds either remedy,
because the normative positions that would license these
answers (Figure~\ref{fig:grounding-reveals}, Row~6), each
provider's Right to Omission and their Power to declare
a violation, are absent from both policy graphs.

What normative positions exist when an ODRL policy is in
force, and what authority structures do those positions
presuppose? The absence of an answer has four concrete
manifestations: \textbf{(P1)}~the \texttt{behaviour} parameter selects between evaluation regimes that create ontologically
incompatible normative positions, with no account of which
positions exist under each setting;
\textbf{(P2)}~violation-declaration authority is absent
from the policy graph and silently delegated to
implementors, causing incompatible runtime determinations
across deployments;
\textbf{(P3)}~every prohibition entails a Right to demand
omission for the assigner, yet neither evaluator nor
policy graph represents it, leaving remedy selection
without a normative basis; and
\textbf{(P4)}~no existing vocabulary enables governance
frameworks to declare sanctioning authority at
policy-authoring time.

We use UFO-L as a foundational ontology to determine 
what kinds of things exist when a policy is in 
force, not to give ODRL a formal-semantics denotation. 
We then \emph{ground} ODRL in 
UFO-L~\cite{griffo2018conceptual,griffo2023powers} to 
address these problems:
P1 by ontologically characterising both \texttt{behaviour}
regimes; P2 and P4 by
constituting a Power--Subjection pair at activation; and P3 by surfacing
the assigner's Right to Omission as an entailed
correlative, extending evaluator coverage from two to eight legal positions.
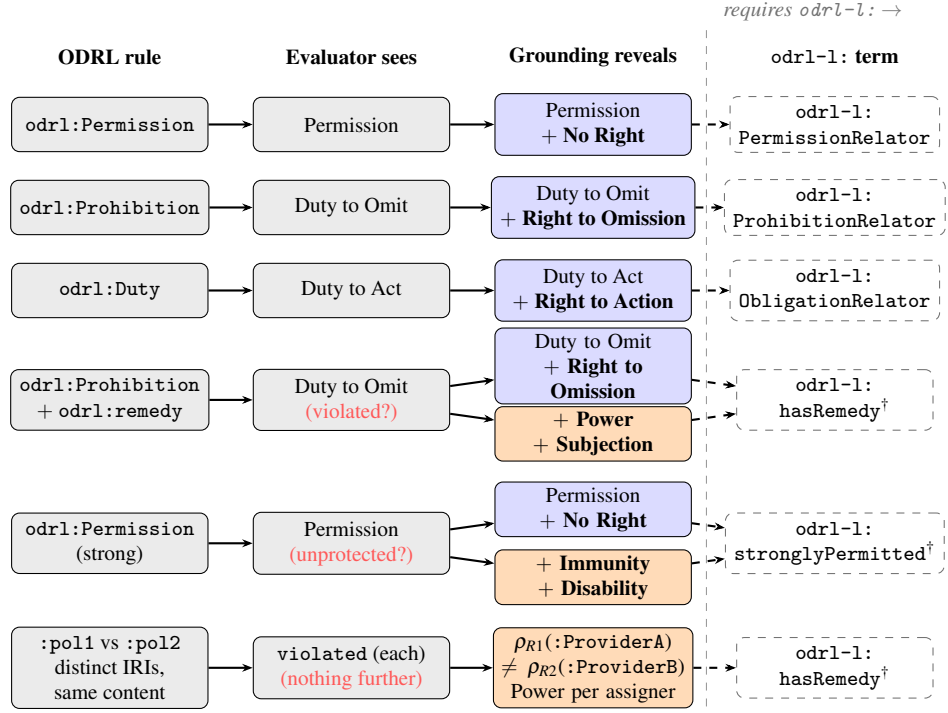
\begin{figure*}[t]
\centering
\resizebox{\textwidth}{!}{%
\begin{tikzpicture}[
  box/.style={draw,rounded corners=3pt,align=center,
    minimum width=2.6cm,minimum height=0.72cm,font=\scriptsize},
  splitbox/.style={box,text width=2.45cm,inner sep=2pt,minimum height=0.62cm},
  widebox/.style={box,text width=2.5cm,inner sep=2pt},
  nativebox/.style={box,fill=black!8},
  conductbox/.style={box,fill=blue!14},
  compbox/.style={box,fill=orange!28},
  profilebox/.style={box,fill=white,draw=black!55,dashed,font=\scriptsize\ttfamily},
  arrow/.style={-{Stealth[length=4pt]},thick},
  darrow/.style={arrow,dashed}
]
\foreach \x/\t in {1.3/{ODRL rule},4.5/{Evaluator sees},%
                   7.7/{Grounding reveals},10.9/{\texttt{odrl-l:} term}}
  \node[font=\scriptsize\bfseries,align=center] at (\x,0) {\t};

\node[nativebox]  (p1) at (1.3,-0.9)  {\odrl{Permission}};
\node[nativebox]  (p2) at (4.5,-0.9)  {Permission};
\node[conductbox] (p3) at (7.7,-0.9)  {Permission\\$+$ \textbf{No Right}};
\node[profilebox] (p4) at (10.9,-0.9) {odrl-l:\\PermissionRelator};
\draw[arrow](p1)--(p2); \draw[arrow](p2)--(p3); \draw[darrow](p3)--(p4);

\node[nativebox]  (f1) at (1.3,-2.0)  {\odrl{Prohibition}};
\node[nativebox]  (f2) at (4.5,-2.0)  {Duty to Omit};
\node[conductbox] (f3) at (7.7,-2.0)  {Duty to Omit\\$+$ \textbf{Right to Omission}};
\node[profilebox] (f4) at (10.9,-2.0) {odrl-l:\\ProhibitionRelator};
\draw[arrow](f1)--(f2); \draw[arrow](f2)--(f3); \draw[darrow](f3)--(f4);

\node[nativebox]  (d1) at (1.3,-3.1)  {\odrl{Duty}};
\node[nativebox]  (d2) at (4.5,-3.1)  {Duty to Act};
\node[conductbox] (d3) at (7.7,-3.1)  {Duty to Act\\$+$ \textbf{Right to Action}};
\node[profilebox] (d4) at (10.9,-3.1) {odrl-l:\\ObligationRelator};
\draw[arrow](d1)--(d2); \draw[arrow](d2)--(d3); \draw[darrow](d3)--(d4);

\node[nativebox] (r1) at (1.3,-4.55) {\odrl{Prohibition}\\$+$ \odrl{remedy}};
\node[nativebox] (r2) at (4.5,-4.55) {Duty to Omit\\\textcolor{red!65}{(violated?)}};
\node[conductbox,splitbox] (r3a) at (7.7,-4.10) {Duty to Omit\\$+$ \textbf{Right to Omission}};
\node[compbox,splitbox]    (r3b) at (7.7,-5.00) {$+$ \textbf{Power}\\$+$ \textbf{Subjection}};
\node[profilebox] (r4) at (10.9,-4.55) {odrl-l:\\hasRemedy$^\dagger$};
\draw[arrow](r1)--(r2);
\draw[arrow](r2)--(r3a); \draw[arrow](r2)--(r3b);
\draw[darrow](r3a)--(r4); \draw[darrow](r3b)--(r4);

\node[nativebox] (s1) at (1.3,-6.45) {\odrl{Permission}\\(strong)};
\node[nativebox] (s2) at (4.5,-6.45) {Permission\\\textcolor{red!65}{(unprotected?)}};
\node[conductbox,splitbox] (s3a) at (7.7,-6.00) {Permission\\$+$ \textbf{No Right}};
\node[compbox,splitbox]    (s3b) at (7.7,-6.90) {$+$ \textbf{Immunity}\\$+$ \textbf{Disability}};
\node[profilebox] (s4) at (10.9,-6.45) {odrl-l:\\stronglyPermitted$^\dagger$};
\draw[arrow](s1)--(s2);
\draw[arrow](s2)--(s3a); \draw[arrow](s2)--(s3b);
\draw[darrow](s3a)--(s4); \draw[darrow](s3b)--(s4);

\node[nativebox] (c1) at (1.3,-8.1)  {\texttt{:pol1} vs \texttt{:pol2}\\distinct IRIs,\\same content};
\node[nativebox] (c2) at (4.5,-8.1)  {\texttt{violated} (each)\\\textcolor{red!65}{(nothing further)}};
\node[compbox,widebox] (c3) at (7.7,-8.1) {$\rho_{R1}$(\texttt{:ProviderA})\\$\neq\;\rho_{R2}$(\texttt{:ProviderB})\\Power per assigner};
\node[profilebox](c4) at (10.9,-8.1) {odrl-l:\\hasRemedy$^\dagger$};
\draw[arrow](c1)--(c2); \draw[arrow](c2)--(c3); \draw[darrow](c3)--(c4);

\draw[black!40,dashed] (9.2,0.35)--(9.2,-8.65);
\node[font=\scriptsize\itshape,black!55,anchor=west] at (9.3,0.6)
  {requires \texttt{odrl-l:} $\rightarrow$};
\end{tikzpicture}%
}
\caption{From an ODRL rule to its UFO-L grounding and the \texttt{odrl-l:}
profile term. Each row reads left to right: the ODRL rule, what a native ODRL
evaluator computes, the full simple legal relator the grounding reveals, and the
\texttt{odrl-l:} term needed to represent it (Table~\ref{tab:positions}).
Shading: \colorbox{black!8}{\strut native} ODRL;
\colorbox{blue!14}{\strut conduct}- and \colorbox{orange!28}{\strut competence}-level
positions; dashed outlines mark the \texttt{odrl-l:} profile.
\textbf{Bold} marks positions absent from the evaluator's output;
\textcolor{red!65}{(?)} marks where the evaluator halts; $^\dagger$ marks the two
ODRL extension points. A prohibition with a remedy (row~4) and a strong
permission (row~5) each found a conduct relator and a competence relator at one
activation event. Row~6 instantiates the motivating scenario of
Section~\ref{sec:introduction}: distinct remedy relators per assigner make
violation-declaration authority provider-specific.}
\label{fig:grounding-reveals}
\end{figure*}
\section{Background}
\subsection{Open Digital Rights Language ODRL}
\label{sec:odrl}
An ODRL policy $\pi = \langle P, F, D \rangle$ consists of three disjoint
rule sets: \odrl{Permission} rules~$P$, \odrl{Prohibition} rules~$F$, and
\odrl{Duty} rules~$D$. Each rule $r = \langle a, x, y, t, C \rangle$
specifies an action type~$a$, an assignee~$x$, an assigner~$y$, a
target~$t$, and a constraint set~$C$~\cite{iannella2018odrl}. The
action~$a$ is a \emph{type} of act rather than a token, since granting a
normative position over a single past token event would be ontologically
incoherent: a policy regulates a kind of act, not one completed
occurrence of it. The target~$t$, by contrast, is a \emph{particular}
resource (a token), while the assignee~$x$ and assigner~$y$ may be either
particular agents (tokens, as in an \odrl{Agreement}) or role types (as in
an \odrl{Offer}).

ODRL defines three 
\odrl{Policy} subclasses by the degree to which parties are 
identified~\cite{iannella2018odrl}. An \odrl{Agreement} 
binds two identified particular agents, creating a bilateral 
normative bond. An \odrl{Offer} is binding \emph{erga omnes} (toward all): 
the assigner is identified while the assignee is the community 
satisfying the offer's conditions~\cite{griffo2021service}; it must not be evaluated by an ODRL Evaluator~\cite{odrl-formal-semantics}. An \odrl{Set} 
is type-level: it creates no binding bond and falls outside 
the scope of this grounding~\cite{iannella2018odrl}. \odrl{Permission} allows an
action when all constraints are satisfied. \odrl{Prohibition}
disallows an action but remains violable: the prohibited action
is performable and its performance triggers
\texttt{violated}~\cite{odrl-formal-semantics}. \odrl{Duty} obliges an
action: \texttt{fulfilled} upon performance and
\texttt{violated} when performance is no longer
possible~\cite{odrl-formal-semantics}. Duties appear in four structural roles: policy-level obligations (\odrl{obligation}), permission conditions (\odrl{duty}),
prohibition remedies (\odrl{remedy}), and failure consequences
(\odrl{consequence})~\cite{iannella2018odrl}. The remedy and consequence roles form a cascade defined in the 
information model~\cite{iannella2018odrl}: if the prohibition is 
infringed, all remedy Duties MUST be fulfilled; if a remedy Duty 
is not fulfilled, its consequence Duties become active. The Formal Semantics defines no activation
procedure for either~\cite{odrl-formal-semantics}.
When a Permission and Prohibition co-apply, the \odrl{conflict}
property selects a resolution strategy: \texttt{perm},
\texttt{prohibit}, or \texttt{invalid}
(default)~\cite{iannella2018odrl}. 
An ODRL rule \emph{prescribes}; a legal position \emph{exists}.
\odrl{Permission}, \odrl{Prohibition}, and \odrl{Duty} are policy-graph nodes, authored by a policy writer, evaluated by
an ODRL Evaluator. The legal positions these rules generate are ontological
individuals that inhere in the parties at activation: they are
not verdicts but normative facts about who may, must, or can do what, and in relation to whom.

\noindent The Formal Semantics defines \texttt{behaviour} as an optional
evaluator input~\cite{odrl-formal-semantics} governing how
unmentioned actions are treated: under \texttt{open}, silence
implies permission; under \texttt{closed} (the default), silence implies prohibition. Its ontological significance is that these two settings create incompatible normative positions. 

\noindent ODRL Profiles~\cite{iannella2018odrl} extend the vocabulary with new terms
(e.g.\ Rule subclasses) defined on top of the Core Vocabulary; a Policy
names the Profile(s) it uses via \texttt{profile}.
\subsection{UFO-L: Foundational Ontology of Legal Positions}
\label{sec:ufol}
UFO-L~\cite{griffo2018conceptual,griffo2023powers} extends the Unified
Foundational Ontology~\cite{Guizzardi2022-PORUUF} with a formal theory of
legal concepts based on Alexy's account of legal positions~\cite{alexy2010theory}. A \emph{moment} is an existentially dependent particular inhering in
exactly one bearer; by the non-migration
principle~\cite[a67]{Guizzardi2022-PORUUF} it cannot migrate between bearers. An \emph{externally dependent moment} inheres in one entity
while remaining existentially dependent on a mereologically disjoint entity. A \emph{relator} is a mereological aggregate of externally
dependent moments~\cite{griffo2021service}: each constituent moment inheres in
its own bearer while depending on the counterparty, and the relator is their mereological sum, capturing the bilateral normative bond as a single ontological individual, founded by a unique
event~\cite[a77]{Guizzardi2022-PORUUF}. \emph{Mediation} is a derived notion, not a primitive: a relator
mediates two entities precisely when it has parts inhering in each
of them~\cite[p.~240]{Guizzardi2005ofscm}.
 
\noindent Legal positions are moments borne by agents, classified as
\emph{Legal Entitlements} (positions implying an advantage) or
\emph{Legal Burdens/Lacks} (positions implying a
constraint)~\cite{griffo2021service}. Positions are always paired as
correlatives within a \emph{simple legal relator} bundling exactly one
such pair~\cite{griffo2021service}: at \emph{conduct level} 
(governing what agents may, must, or must not do) the 
four types are Permission--No Right and Right--Duty; 
at \emph{competence level} (governing the authority to 
bring normative positions into or out of existence through 
institutional acts~\cite{alexy2010theory,griffo2021service}) they are Power--Subjection and Immunity--Disability. Their mapping to ODRL is shown in 
Table~\ref{tab:positions}. Disability denies Power, rendering
institutional acts void~\cite{griffo2023powers}. Immunity can arise from any source that constitutionally
entrenches a position against revocation~\cite{griffo2023powers}. Legal relators are always brought into existence by a competent
agent's Power~\cite{griffo2023powers}, motivating the cross-level
design principle.
\vspace{-15pt}
\section{Deontic Analysis of ODRL: Permission, Prohibition, and Duty}
\label{sec:classification}
Permission, Prohibition, and Duty are the only ODRL constructs
that generate normative positions between parties; Asset and Party
are structural \emph{relata}, and Constraint an \emph{activation
condition} that determines when a rule applies but creates no legal
position~\cite{iannella2018odrl}.
\subsection{Permission: Underspecified Between Weak and Strong}
\label{sec:classification:permission}
\emph{Strong permission} persists under policy modification
and requires an Immunity blocking any subsequent prohibition, a
competence-level position~\cite{griffo2023powers,alexy2010theory}.
In UFO-L terms, Permission to Act and Duty to Omit are
\emph{opposites}: they cannot be co-borne by the same agent over
the same action and target. Under \texttt{behaviour=open},
undeclared actions carry no positive normative position: the
evaluator's default is a convention, not a granted Permission.
Under \texttt{behaviour=closed}, each declared permission creates
a Permission--No Right relator. However, the Formal Semantics
defines no persistence protection~\cite{odrl-formal-semantics}:
closed-world evaluation is consistent with but not equivalent to strong permission. The \texttt{behaviour} parameter therefore selects between ontologically incompatible normative positions, a distinction invisible to the evaluator but explicit in the UFO-L account (P1).
\subsection{Prohibition: Provably Sanctioned}
\label{sec:classification:prohibition}
Prohibitions are either \emph{regimented}, violation structurally
impossible or \emph{sanctioned}, violation possible but
consequential.
\begin{proposition}
\label{prop:sanctioned}
ODRL prohibition is \emph{sanctioned}: the prohibited action remains
performable and its performance triggers normative consequences.
\end{proposition}
\begin{proof}[Proof]
Assume regimentation. Then ~\cite{iannella2018odrl}, `if the
Prohibition has been \emph{infringed by the action being exercised},
then all remedies MUST be fulfilled'', describes an unreachable
scenario, rendering \texttt{odrl:remedy} permanently inactivatable.
Yet the Recommendation defines, exemplifies, and requires implementors
to support it~\cite{iannella2018odrl}. Contradiction. The Formal
Semantics confirms independently via an explicit \texttt{violated}
state~\cite{odrl-formal-semantics}. Since prohibition is sanctioned and violation triggers a
reparative duty, any complete grounding requires a competence-level
Power--Subjection pair to ground the violation-to-remedy
transition~(P2).
\end{proof}

\subsection{Duty: Achievement Semantics Only}
\label{sec:classification:duty}
Obligations divide into three types by their
satisfier~\cite{wright1963norm,guarino2024processes}:
\emph{achievement} (satisfied by a bounded event),
\emph{maintenance} (satisfied by an ongoing state), and
\emph{process} (satisfied by a continuous activity). Events cannot
be ongoing, since being ongoing is a temporary property incompatible
with the timeless character of events~\cite{guarino2024processes}.
\begin{proposition}
\label{prop:achievement}
The ODRL Formal Semantics covers achievement obligations only;
maintenance and process obligations lack evaluation semantics.
\end{proposition}
\begin{proof}[Proof sketch]
The three-state semantics, \texttt{fulfilled} upon a single
performance event, \texttt{violated} when performance is no longer
possible, \texttt{not-set} otherwise~\cite{odrl-formal-semantics}, presupposes
an event-bounded satisfier. Maintenance and process obligations have
no such satisfier: the evaluator defines no transition for a
persisting state or an ongoing activity.
\end{proof}
\vspace{-15pt}
\section{Cross-Level Design Principle}
\label{sec:crosslevel}
Any complete ontological grounding of ODRL requires
competence-level legal positions. All results are conditional on three
axioms formalising UFO-L~\cite{griffo2023powers}. \textbf{A1}~normative state changes require a
triggering institutional event; \textbf{A2}~such events require a
competent agent; \textbf{A3}~competence is a Power--Subjection pair.
These axioms hold under both regulative and constitutive readings of norm
activation: even constitutive rules presuppose a prior creation event
grounded in an agent's Power~\cite{griffo2023powers}.
\begin{definition}[Level-Complete and Cross-Level]
\label{def:level-complete}
Let $N$ be a normative concept and $L$ a level of legal
positions (either conduct or competence).
A set of legal positions $\mathcal{H}$ is an \emph{adequate
characterization} of $N$ relative to axiom set
$\mathcal{A}$ if $\mathcal{H}$ entails all normative consequences of $N$
under any model extension consistent with $\mathcal{A}$.
A \emph{normative consequence} of $N$ is a legal position entailed by
$N$; quantifying over model extensions makes adequacy robust, requiring
$\mathcal{H}$ to keep entailing those positions as further rules are
added, not only in one fixed model.
$N$ is \emph{level-complete} at $L$ if it has an adequate
characterization using only positions from $L$.
$N$ is \emph{cross-level} if every adequate characterization
requires positions from both conduct and competence levels.
\end{definition}
\begin{proposition}[Weak Permission is Conduct-Level-Complete]
\label{prop:weak-complete}
Relative to \textbf{A1}--\textbf{A3}, weak permission for an assignee
$x$ to perform an action $a$ on a target $t$, granted by an assigner
$y$, is adequately characterised at the conduct level by
$\{\PermPos(x,a,t),\NoRight(y,a,t)\}$ (here $x$ bears a Permission to
Act and $y$ the correlative No Right over $a$ on $t$): no
competence-level positions are required. This establishes that the cross-level requirement is non-trivial: only violable and consequential concepts require competence-level positions.
\end{proposition}
\begin{theorem}[Strong Permission is Cross-Level]
\label{thm:strong-crosslevel}
Relative to \textbf{A1}--\textbf{A3}, strong permission of $x$ to perform $a$ on $t$, granted by $y$, has no adequate characterization
using conduct-level positions alone.
\end{theorem}
\noindent
\begin{theorem}[Sanctioned Prohibition is Cross-Level]
\label{thm:sanctioned-crosslevel}
Relative to \textbf{A1}--\textbf{A3}, sanctioned prohibition has no
adequate characterization using conduct-level positions alone. A norm is \emph{violable} if the regulated action
remains performable despite the norm; it is
\emph{consequential} if violation activates a reparative
normative position.
\end{theorem}
\noindent
\begin{theorem}[Violable Norms Require Both Levels]
\label{thm:crosslevel}
Relative to \textbf{A1}--\textbf{A3}, any norm that is
(i)~violable and (ii)~consequential has no adequate
characterization using conduct-level positions alone.
\end{theorem}
\noindent
\begin{principle}[Cross-Level Design]
\label{prin:crosslevel}
Any normative language with violable, consequential norms must
provide competence-level positions (Power, Subjection, Immunity,
Disability) alongside conduct-level ones (Permission, Duty,
Right, No right): a conduct-only vocabulary cannot ground
violation-to-consequence transitions regardless of implementation.
\end{principle}
\noindent Since ODRL prohibition is sanctioned
(Proposition~\ref{prop:sanctioned}) and remedies are normative consequences of
violation by definition~\cite{iannella2018odrl}, Theorem~\ref{thm:crosslevel}
entails that any complete grounding of ODRL requires competence-level positions.
The Berliner Ensemble prohibition of Section~\ref{sec:introduction} is exactly
such a norm: its remedy is enforceable only if some party holds the Power to
declare the violation, a competence-level position that no conduct-only
vocabulary supplies (grounded concretely in Example~\ref{ex:pol1-relators}).
\vspace{-15pt}
\section{Formal Grounding of ODRL in UFO-L}
\label{sec:grounding}
Each activated ODRL rule maps to a \emph{simple legal relator}
(Section~\ref{sec:ufol}), the minimal adequate grounding required
by Theorem~\ref{thm:crosslevel}: a single rule creates exactly one
correlative pair of legal positions, and a simple legal relator
bundles exactly one such pair~\cite{griffo2023powers}. A
\emph{forbearance} $\mathit{rfr}(a)$ is the \emph{omission} of
action $a$: in UFO-L, refraining from performing $a$ is a perdurant
of kind omission, distinct from act $a$~\cite{griffo2018conceptual}.
Accordingly, a Duty to refrain from $a$ is a Duty to Omit $a$, and
its correlative is a Right to Omission of $a$. We use
$\Forbearance$ to denote the sort of all omissions; $\mathit{Act}$
and $\Forbearance$ are disjoint sorts with
$\mathit{rfr}(a)\in\Forbearance$ for every $a\in\mathit{Act}$, and
the map $\mathit{rfr}$ is injective. We write $\mathit{decl}(a)$ for
the institutional act of declaring a violation of~$a$;
$\mathit{decl}$ maps $\mathit{Act}$ to $\mathit{Act}$ injectively.
Throughout, the eight legal positions are written in a three-place
form: $\PermPos(x,a,t)$ reads ``$x$ bears a Permission to Act over
action~$a$ on target~$t$'', and likewise for $\Duty$, $\Right$,
$\NoRight$, $\Power$, $\Subj$, $\Immunity$, and $\Disability$.

\noindent Each active rule becomes a two-party legal relationship: one
position for the assignee and the matching correlative for the assigner.
\begin{definition}[Rules as Simple Legal Relators]
\label{def:rule-relator}
{\small
\[
\begin{aligned}[t]
\textbf{Permission } \rho_P &:~
\PermPos(x,a,t)\land\NoRight(y,a,t),\quad a\in\mathit{Act}\\
\textbf{Strong permission } \rho_I &:~
\Immunity(x,a,t)\land\Disability(y,a,t)\quad[\text{if }\mathit{strong}(p)]\\
\textbf{Prohibition } \rho_F &:~
\Duty(x,\mathit{rfr}(a),t)\land\Right(y,\mathit{rfr}(a),t),\quad \mathit{rfr}(a)\in\Forbearance\\
\textbf{Prohibition remedy } \rho_R &:~
\Power(y,\mathit{decl}(a),t)\land\Subj(x,\mathit{decl}(a),t)\quad[\text{if }\hasrem(f)]\\
\textbf{Duty } \rho_D &:~
\Duty(x,a,t)\land\Right(y,a,t),\quad a\in\mathit{Act}
\end{aligned}
\]}
Each rule founds exactly one relator via $\mathit{founds}$ (the conduct relators
$\rho_P,\rho_F,\rho_D$); additionally, a prohibition with \odrl{remedy} founds
$\rho_R$ via $\mathit{founds\_rem}$, and a strong permission founds $\rho_I$ via $\mathit{founds\_imm}$, both at the same activation event. The distinct
predicates keep the two relators separate.
\end{definition}
\noindent
\noindent
\begin{table}[t]
\centering\scriptsize
\setlength{\tabcolsep}{2pt}
\caption{ODRL rules, UFO-L relators, and legal positions surfaced
  at activation~\cite{griffo2023powers}.
  $\bullet$~correlative; $\blacksquare$~added by grounding;
  $\dagger$~profile extension (\texttt{odrl-l:});
  \textsc{ee}~=~assignee; \textsc{er}~=~assigner.}
\label{tab:positions}
\begin{tabular}{@{}lllcc@{}}
\toprule
\textbf{ODRL rule} & \textbf{Relator} &
\textbf{Position} & \textbf{Bearer} & \textbf{Level} \\
\midrule
\multirow{2}{*}{\odrl{Perm.} (weak)}
  & \multirow{2}{*}{$\rho_P$}
  & \cellcolor{blue!14}Permission to Act     & \textsc{ee} & conduct \\
  && \cellcolor{blue!14}No Right~$\bullet$    & \textsc{er} & conduct \\
\addlinespace[2pt]
\multirow{2}{*}{\odrl{Perm.} (strong$^\dagger$)}
  & \multirow{2}{*}{$+\rho_I$}
  & \cellcolor{orange!28}Immunity~$\blacksquare\dagger$    & \textsc{ee} & competence \\
  && \cellcolor{orange!28}Disability~$\blacksquare\dagger$ & \textsc{er} & competence \\
\addlinespace[2pt]
\multirow{2}{*}{\odrl{Prohibition}}
  & \multirow{2}{*}{$\rho_F$}
  & \cellcolor{blue!14}Duty to Omit          & \textsc{ee} & conduct \\
  && \cellcolor{blue!14}Right to\ Omission~$\bullet$ & \textsc{er} & conduct \\
\addlinespace[2pt]
\multirow{2}{*}{\odrl{Proh.}+remedy$^\dagger$}
  & \multirow{2}{*}{$+\rho_R$}
  & \cellcolor{orange!28}Power~$\blacksquare\dagger$       & \textsc{er} & competence \\
  && \cellcolor{orange!28}Subjection~$\blacksquare\dagger$ & \textsc{ee} & competence \\
\addlinespace[2pt]
\multirow{2}{*}{\odrl{Duty}}
  & \multirow{2}{*}{$\rho_D$}
  & \cellcolor{blue!14}Duty to Act           & \textsc{ee} & conduct \\
  && \cellcolor{blue!14}Right to\ Action~$\bullet$   & \textsc{er} & conduct \\
\bottomrule
\multicolumn{5}{@{}p{\columnwidth}}{\tiny
  Strong perm.\ founds $\rho_P$ (weak row) \emph{and}
  $\rho_I$ at same event (Ax5.2).
  $\dagger$\texttt{odrl-l:hasRemedy} (Ax5.4);
  \texttt{odrl-l:stronglyPermitted} (Ax5.2).
  Immunity not exclusive to strong permission~\cite{griffo2023powers}.
  Founding: $\rho_P,\rho_F,\rho_D$ via $\mathit{founds}$;
  $\rho_I$ via $\mathit{founds\_imm}$;
  $\rho_R$ via $\mathit{founds\_rem}$.}
\end{tabular}
\end{table}
\noindent
\vspace{15pt}
\subsection{Resolving Underspecifications}
\label{sec:grounding:resolution}
\begin{proposition}[Open-World Normative Vacuum]
\label{prop:open-vacuum}
Under \texttt{behaviour=open}, for any action $a$ not regulated
by a rule in $\pi$, the position set $\mathit{Pos}(\pi,e)$
(Definition~\ref{def:ground}) contains no Permission,
No Right, Duty, or Right over $\langle a,t\rangle$ for any bearer.
\end{proposition}
\paragraph{P1: Behaviour parameter.}
Proposition~\ref{prop:open-vacuum} establishes that under
\texttt{open} the open-world default is a runtime convention
with no ontological counterpart in $\mathcal{A}$.
Under \texttt{closed}, each declared permission creates relator
$\rho_P$; adding $\Immunity(x,a,t)\land\Disability(y,a,t)$
realises strong permission persisting under subsequent normative
change (Theorem~\ref{thm:strong-crosslevel}). The
\texttt{behaviour} parameter is thus a choice between
ontologically incompatible normative positions, not merely a
configuration switch.

\paragraph{P2: Violation authority.}
For every prohibition with a remedy, the grounding adds
$\Power(y,\mathit{decl}(a),t)\land\Subj(x,\mathit{decl}(a),t)$,
surfacing the evaluator's implicit role as normative authority. The Power constituted at activation licenses $y$ not only to declare the violation
directly but also to \emph{delegate} sanctioning competence to a third
party, creating a new Power--Subjection pair in another agent via a
normative transfer act~\cite{griffo2023powers}.
Governance frameworks can thereby specify \emph{which} party holds
violation-declaration authority for \emph{which} prohibitions
(Theorem~\ref{thm:sanctioned-crosslevel}).
\paragraph{P3: Assigner correlatives.}
The grounding surfaces a Right to Omission for every assigner
of a prohibition and a Right to an Action for every assigner of a
duty (Table~\ref{tab:positions}, Axiom~\ref{ax:proh-relator-conduct}).
These correlatives are ontologically entailed but absent from the
ODRL evaluator: without them, no reasoner can identify whose
normative standing grounds a given remedy, even when dataset
identity is known.
\paragraph{P4: Normative authority.}
The Power--Subjection pair constituted at activation
(Axiom~\ref{ax:proh-relator-remedy}) makes violation-declaration
authority explicit, enabling data sovereignty representation
in governance frameworks.
\subsection{First-Order Axiomatization}
\label{sec:grounding:axioms}
Each axiom answers the same question: when an ODRL rule activates, what legal individuals come into existence? The answer is always a simple legal relator that bundles one correlative pair, one position
per party, founded by the activation event and aggregated via $\mathit{partOf}$.
\noindent\textit{Predicate key.}
\emph{ODRL rule (policy graph):}
$\Perm(p)$/$\Proh(f)$/$\mathit{obl}(d)$~=~Permission/Prohibition/Duty
rule node ($\mathit{obl}$, the Duty \emph{rule}, is distinct from
$\Duty$, the UFO-L Duty \emph{position});
$\mathit{aee}(r,x)$/$\mathit{aer}(r,y)$~=~$r$ has assignee~$x$ / assigner~$y$;
$\mathit{act}(r,a)$~=~$r$ regulates action type~$a$;
$\mathit{tgt}(r,t)$~=~$r$ has target token~$t$;
$\hasrem(f)$~=~$f$ carries \odrl{remedy};
$\mathit{strong}(p)$~=~$p$ is strongly permitted (profile extension,
\S\ref{sec:completeness:profile}).
\emph{Activation and founding (events):}
$\activates(e,r)$~=~$e$ activates~$r$~\cite{odrl-formal-semantics};
$\founds$/$\mathit{founds\_rem}$/$\mathit{founds\_imm}$~=~conduct / remedy /
immunity founding of a relator by a rule at an event.
\emph{UFO-L relator and moments (ontology):}
$\mathit{Rel}(\rho)$~=~$\rho$ is a UFO-L legal relator;
$\mathit{ODRLRel}(\rho)$~=~$\rho$ is a relator founded by an ODRL rule;
$\mathit{bearer}(m,x)$~=~moment~$m$ inheres in~$x$;
$\mathit{partOf}(m,\rho)$~=~$m$ is part of relator~$\rho$;
$\mathit{cnt}(m,a,t)$~=~$m$ has normative content over action type~$a$
and target~$t$.
\noindent We inherit from UFO the \emph{relator mediation} relation:
a relator $\rho$ mediates $x$ and $y$ iff it has parts inhering in $x$ and $y$ respectively~\cite[p.~240]{Guizzardi2005ofscm}, derived from
$\mathit{partOf}$ and $\mathit{bearer}$, not re-axiomatised here.

\noindent Activating a permission creates a relator in which the
assignee holds a Permission to Act and the assigner the correlative No Right,
both over the same action and target.
\begin{axiom}[Permission Relator Weak]
\label{ax:perm-relator-weak}
\[
\begin{aligned}[t]
\forall p,x,y,a,t,e.\;&
\Perm(p)\land\mathit{aee}(p,x)\land\mathit{aer}(p,y)\land
\mathit{act}(p,a)\land\mathit{tgt}(p,t)\land\activates(e,p)\to\\
\exists\rho_P,l,n.\;&
\mathit{Rel}(\rho_P)\land\founds(e,\rho_P,p)\land
\PermPos(l)\land\mathit{bearer}(l,x)\land
\mathit{cnt}(l,a,t)\\
&\land\mathit{partOf}(l,\rho_P)\land
\NoRight(n)\land\mathit{bearer}(n,y)\land
\mathit{cnt}(n,a,t)\land\mathit{partOf}(n,\rho_P)
\end{aligned}
\]
\end{axiom}
\noindent A strong permission also founds a second relator at the same
activation event: besides the permission relator $\rho_P$ (Permission--No Right,
Axiom~\ref{ax:perm-relator-weak}), it founds the immunity relator $\rho_I$,
which gives the assignee an Immunity and the assigner the correlative Disability.
\begin{axiom}[Permission Relator Strong]
\label{ax:perm-relator-strong}
\[
\begin{aligned}[t]
\forall p,x,y,a,t,e.\;&
\Perm(p)\land\mathit{strong}(p)\land\mathit{aee}(p,x)\land
\mathit{aer}(p,y)\land\mathit{act}(p,a)\land\mathit{tgt}(p,t)\land\activates(e,p)\to\\
\exists\rho_I,im,db.\;&
\mathit{Rel}(\rho_I)\land\mathit{founds\_imm}(e,\rho_I,p)\land
\Immunity(im)\land\mathit{bearer}(im,x)\land\mathit{cnt}(im,a,t)\land\\
&\mathit{partOf}(im,\rho_I)\land
\Disability(db)\land\mathit{bearer}(db,y)\land
\mathit{cnt}(db,a,t)\land\mathit{partOf}(db,\rho_I)
\end{aligned}
\]
\end{axiom}
\noindent Activating a prohibition creates the prohibition relator $\rho_F$, in
which the assignee owes a Duty to Omit the action $a$ (its content is the
omission $\mathit{rfr}(a)$, not $a$) and the assigner holds the correlative
Right to that omission.
\begin{axiom}[Prohibition Relator---Conduct]
\label{ax:proh-relator-conduct}
\[
\begin{aligned}[t]
\forall f,x,y,a,t,e.\;&
\Proh(f)\land\mathit{aee}(f,x)\land\mathit{aer}(f,y)\land
\mathit{act}(f,a)\land\mathit{tgt}(f,t)\land\activates(e,f)\to\\
\exists\rho_F,d,c.\;&
\mathit{Rel}(\rho_F)\land\founds(e,\rho_F,f)\land
\Duty(d)\land\mathit{bearer}(d,x)\land\mathit{cnt}(d,\mathit{rfr}(a),t)\land\\
&\mathit{partOf}(d,\rho_F)\land
\Right(c)\land\mathit{bearer}(c,y)\land
\mathit{cnt}(c,\mathit{rfr}(a),t)\land\mathit{partOf}(c,\rho_F)
\end{aligned}
\]
\end{axiom}

\noindent A prohibition with a remedy founds a competence-level relator $\rho_R$: the assigner gains a standing Power to declare a violation (over
$\mathit{decl}(a)$) and the assignee the correlative Subjection.
\begin{axiom}[Prohibition Relator---Remedy]
\label{ax:proh-relator-remedy}
\begin{align*}
\forall f,x,y,a,t,e.\;
  &\Proh(f)\land\hasrem(f)\land\mathit{aee}(f,x)\land\mathit{aer}(f,y)\land{}\\
  &\mathit{act}(f,a)\land\mathit{tgt}(f,t)\land\activates(e,f)\to{}\\
  &\exists\rho_R,pw,s.\;
   \mathit{Rel}(\rho_R)\land\mathit{founds\_rem}(e,\rho_R,f)\land{}\\
  &\quad\Power(pw)\land\mathit{bearer}(pw,y)\land
   \mathit{cnt}(pw,\mathit{decl}(a),t)\land\mathit{partOf}(pw,\rho_R)\land{}\\
  &\quad\Subj(s)\land\mathit{bearer}(s,x)\land
   \mathit{cnt}(s,\mathit{decl}(a),t)\land\mathit{partOf}(s,\rho_R)
\end{align*}
\end{axiom}
\noindent Activating an obligation creates the duty relator $\rho_D$, in which
the assignee owes a Duty to Act on $a$ and the assigner holds the correlative
Right to that action.
\begin{axiom}[Obligation Relator]
\label{ax:obl-relator}
\[
\begin{aligned}[t]
\forall d,x,y,a,t,e.\;&
\mathit{obl}(d)\land\mathit{aee}(d,x)\land\mathit{aer}(d,y)\land
\mathit{act}(d,a)\land\mathit{tgt}(d,t)\land\activates(e,d)\to\\
\exists\rho_D,du,c.\;&
\mathit{Rel}(\rho_D)\land\founds(e,\rho_D,d)\land
\Duty(du)\land\mathit{bearer}(du,x)\land\mathit{cnt}(du,a,t)\land\\
&\mathit{partOf}(du,\rho_D)\land
\Right(c)\land\mathit{bearer}(c,y)\land\mathit{cnt}(c,a,t)\land\mathit{partOf}(c,\rho_D)
\end{aligned}
\]
\end{axiom}
\noindent For each founding mode a rule founds one relator at a given event (and each relator has one founding event); since the modes differ, founding in
two modes yields two separate relators.
\begin{axiom}[Unique Founding]
\label{ax:unique-founding}
For $\mathit{fd}\in\{\founds,\mathit{founds\_rem},\mathit{founds\_imm}\}$:
\[
\begin{aligned}[t]
&\forall r,\rho_1,\rho_2,e.\;
\mathit{fd}(e,\rho_1,r)\land\mathit{fd}(e,\rho_2,r)\to\rho_1=\rho_2\\
&\forall r,e_1,e_2,\rho.\;
\mathit{fd}(e_1,\rho,r)\land\mathit{fd}(e_2,\rho,r)\to e_1=e_2
\end{aligned}
\]
\textit{(Non-migration and temporal non-overlap.)}
\end{axiom}
\noindent A relator is an ODRL relator when a rule of the matching kind founds
it: any rule via $\mathit{founds}$, a prohibition via $\mathit{founds\_rem}$, a
permission via $\mathit{founds\_imm}$.
\begin{axiom}[ODRL Relator Typing]
\label{ax:odrl-rel-typing}
\[
\forall e,\rho,r.\;
\mathit{fd}(e,\rho,r)\land G(r)\to\mathit{ODRLRel}(\rho)
\]
where $(\mathit{fd},G)\in
\{(\founds,\,\Perm\lor\Proh\lor\mathit{obl}),\;
(\mathit{founds\_rem},\,\Proh),\;
(\mathit{founds\_imm},\,\Perm)\}$.
\end{axiom}
 \begin{axiom}[ODRL Relator Is a Relator]
\label{ax:odrl-rel-is-rel}
\[
\forall\rho.\;\mathit{ODRLRel}(\rho)\to\mathit{Rel}(\rho)
\]
\end{axiom}
\noindent Within an ODRL relator each correlative pair comes as a matched unit: one side is present, and unique, exactly when its correlative is, so a relator never holds half a pair or a duplicate.
\begin{axiom}[Correlativity]
\label{ax:correlativity}
For each correlative pair $(X,Y)\in\{(\PermPos,\NoRight),(\Duty,\Right),
(\Power,\Subj),(\Immunity,\Disability)\}$, every ODRL simple legal relator
$\rho$ contains exactly one position from that pair over the same
content~\cite{griffo2023powers}. For all $\rho,a,t$ with
$\mathit{ODRLRel}(\rho)$:
{\scriptsize
\[
\begin{aligned}
&(\exists!\,u.\;X(u)\land\mathit{partOf}(u,\rho)\land\mathit{cnt}(u,a,t))\\
&\quad\leftrightarrow
  (\exists!\,v.\;Y(v)\land\mathit{partOf}(v,\rho)\land\mathit{cnt}(v,a,t))
\end{aligned}
\]}
\end{axiom}
\begin{axiom}[Normative Position Incompatibility]
\label{ax:cross-relator}
No agent can simultaneously bear a Permission to Act and a Duty to Omit
over the same $\langle a,t\rangle$. This is a \emph{normative} fact
independent of UFO type disjointness, which governs whether a single
moment individual can have two types, not whether a bearer can hold two
distinct moments of incompatible types:
\[
\begin{aligned}[t]
\forall l,d,x,a,t.\;
&\PermPos(l)\land\mathit{bearer}(l,x)\land\mathit{cnt}(l,a,t)\land{}\\
&\Duty(d)\land\mathit{bearer}(d,x)\land\mathit{cnt}(d,\mathit{rfr}(a),t)
  \to\bot
\end{aligned}
\]
\end{axiom}
\begin{corollary}[Permission--Duty Conflict Within a Relator]
\label{ax:conflict}
No ODRL simple legal relator $\rho$ contains both a Permission to Act and
a Duty to Omit over the same $\langle a,t\rangle$ for the same bearer.
(Immediate from Axiom~\ref{ax:cross-relator}.)
\end{corollary}
\begin{axiom}[Disability Precludes Prohibition Creation]
\label{ax:disability-block}
\begin{multline*}
  \forall f,x,y,a,t.\;
  \Proh(f)\land\mathit{aee}(f,x)\land\mathit{aer}(f,y)
  \land\mathit{act}(f,a)\land\mathit{tgt}(f,t) \\
  \to\neg\exists\mathit{db}.\;
  \Disability(\mathit{db})\land\mathit{bearer}(\mathit{db},y)
  \land\mathit{cnt}(\mathit{db},a,t)
\end{multline*}
\end{axiom}
\noindent\textbf{Predicate shorthand.}
We write $\PermPos(x,a,t)$ for
$\exists l.\,\PermPos(l)\land\mathit{bearer}(l,x)\land\mathit{cnt}(l,a,t)$;
analogously for $\Duty$, $\Right$, $\NoRight$, $\Power$, $\Subj$,
$\Immunity$, $\Disability$.

\begin{definition}[Position Instantiation]
\label{def:ground}
For policy $\pi=\langle P,F,D\rangle$ and activation event $e$, the
\emph{position set} $\mathit{Pos}(\pi,e)$ is the set of legal position
instances entailed by $\mathcal{A}=\{$Axioms~\ref{ax:perm-relator-weak}--\ref{ax:disability-block}$\}$:
\[
  \mathit{Pos}(\pi,e) =
  \{\,q \mid \mathcal{A}\cup
  \{\,\activates(e,r)\mid r\in\pi\,\}
  \models q\,\}
\]
The term \emph{grounding} is reserved for the ontological interpretation
of ODRL in UFO-L; $\mathit{Pos}(\pi,e)$ is the \emph{instantiation}
operation deriving concrete position individuals at a specific activation.
\end{definition}
\begin{proposition}[Evaluator Soundness and Incompleteness]
\label{prop:faithfulness}
Under \texttt{closed}: evaluator permits $a$ for $x$ on $t$
$\Rightarrow$ $\PermPos(x,a,t)\in\mathit{Pos}(\pi,e)$.
Under both settings: prohibition activation $\Rightarrow$
$\Duty(x,\mathit{rfr}(a),t)$; prohibition with remedy $\Rightarrow$
$\Power(y,\mathit{decl}(a),t)$.
The converse fails: $\mathit{Pos}(\pi,e)$ additionally entails No Right,
Right to Omission, Immunity, and Disability (Table~\ref{tab:positions}).
The evaluator is thus \emph{sound} but not \emph{complete}:
four of the eight entailed positions are absent from the
Evaluation Report.
\end{proposition}
\begin{example}[\texttt{:BE-Agreement} grounded]
\label{ex:pol1-relators}
Activating \texttt{:BE-Agreement} (assignee \texttt{:KulturPortal}, assigner
\texttt{:BerlinerEnsemble}, action \odrl{distribute}, target
\texttt{:BE-Showtimes}) founds two relators at one event, kept distinct by
Axiom~\ref{ax:unique-founding}. Axiom~\ref{ax:proh-relator-conduct} founds the
conduct relator $\rho_F$: \texttt{:KulturPortal} a Duty to Omit (content
$\mathit{rfr}(\odrl{distribute})$), \texttt{:BerlinerEnsemble} the correlative
Right to Omission. As it carries a remedy, Axiom~\ref{ax:proh-relator-remedy}
founds the competence relator $\rho_R$: \texttt{:BerlinerEnsemble} a Power over
$\mathit{decl}(\odrl{distribute})$ (to declare a violation),
\texttt{:KulturPortal} the correlative Subjection. \texttt{:DT-Agreement} founds
its own $\rho_R$ for \texttt{:DeutschesTheater}, so the two theaters hold
\emph{distinct} violation-declaration Powers, the authority left implicit by the
evaluator's two \texttt{violated} verdicts (Section~\ref{sec:introduction}).
\end{example}
\subsection{The ODRL Legal Profile}
\label{sec:completeness:profile}
The grounding motivates a companion OWL profile, the \emph{ODRL Legal
Profile} (\texttt{odrl-l:})\footnote{\url{https://w3id.org/odrl-legal/}},
enabling policy authors to assert competence-level positions at
authoring time~(P4). The profile defines 15~classes (1~abstract
position supertype, the 8~leaf classes of Table~\ref{tab:positions},
1~abstract relator root, and 5~concrete relator classes) and 8~properties.
The two ODRL extension points ($\dagger$ in
Table~\ref{tab:positions}) are \texttt{odrl-l:hasRemedy} and
\texttt{odrl-l:stronglyPermitted}. Three OWL approximation limits apply, all stemming from OWL's lack of function
symbols and its decidable description-logic fragment: $\mathit{rfr}(a)$
forbearances are not representable (no term for the omission as content);
\texttt{stronglyPermitted} carries no OWL entailment (an annotation flag, since
founding the immunity relator by Axiom~\ref{ax:perm-relator-strong} is a rule
OWL cannot derive); and Axioms~\ref{ax:correlativity}
and~\ref{ax:cross-relator} require full FOL (a unique-existence biconditional
over shared content, and a contradiction over $\mathit{rfr}$-related content).
\vspace{-15pt}
\section{Evaluation}
\label{sec:validation}
We encode the axioms of Section~\ref{sec:grounding:axioms} in three
independent systems: TPTP/FOF for Vampire~5.0.0~\cite{kovacs2013first}
and E~3.2.5~\cite{schulz2013system}, SMT-LIB~2 (\texttt{UF}) for
Z3~4.15.4~\cite{de2008z3}, and Isabelle/HOL~\cite{nipkow2002isabelle},
which verifies all axioms and the corollary, the sanctioned-prohibition
lifecycle, and Proposition~\ref{prop:faithfulness}. The benchmark has
39~problems over 36~identifiers (GRND001--036), three of which carry two
variants for the open/closed-world, sanctioned/regimented, and
with/without-immunity contrasts. Each problem is emitted as a TPTP/FOF
file, an SMT-LIB~2 file, and a Turtle policy over DRK culture-dataspace
entities~\cite{isprs-archives-XLVIII-M-9-2025-1515-2025}, and inlines
only the axioms it needs: entailment problems confirm those axioms are
sufficient, discriminating problems that they are not too strong.

\noindent\textbf{Consistency.}
The full axiom set is satisfiable, witnessed by a minimal finite model (one \odrl{Permission}, two agents) from Vampire's
finite-model schedule. \textbf{Entailment.} Each problem derives a grounding claim from a strict subset of axioms,
including the three clauses of Proposition~\ref{prop:faithfulness} and
the cross-level chain behind Theorem~\ref{thm:crosslevel}. One
instantiates the motivating scenario of Section~\ref{sec:introduction}:
activating \texttt{drk:BerlinerEnsemble}'s prohibition over
\texttt{drk:TheaterShowtimeDataset} derives exactly the Power that
licenses violation declaration.

\noindent\textbf{Discriminating.}
These witness the ontological distinctions an evaluator cannot express: open- vs.\ closed-world permission, sanctioned vs.\ regimented
prohibition, weak vs.\ strong permission, multi-policy conflict, and the end-to-end remedy chain. All checks pass (Table~\ref{tab:results}); E omits the
3~satisfiable problems (no finite-model finder) and Z3 times out on them
in the \texttt{UF} fragment, while Vampire's portfolio schedule
discharges them.

\begin{table}[t]
\centering\scriptsize
\setlength{\tabcolsep}{3pt}
\caption{Axiom coverage matrix. Each row reports the SZS status of
  problems that exercise that axiom (possibly alongside others).
  $\times n$: the axiom is a schema with $n$~FOL instances, each
  tested by a dedicated problem.
  \textbf{Thm}~=~theorem; \textbf{Uns}~=~unsatisfiable;
  \textbf{Sat}~=~satisfiable~\cite{sutcliffe2008szs}.}
\label{tab:results}
\begin{tabular}{@{}lc@{\quad}lc@{}}
\toprule
\textbf{Axiom (Ref)} & \textbf{SZS} & \textbf{Axiom (Ref)} & \textbf{SZS} \\
\midrule
perm-relator-weak\ \ (Ax\ref{ax:perm-relator-weak})    & Thm     & correlativity ×4\ (Ax\ref{ax:correlativity})   & Thm/Uns \\
perm-relator-strong\ (Ax\ref{ax:perm-relator-strong})    & Thm/Uns & cross-relator\ \ \ (Ax\ref{ax:cross-relator}) & Uns     \\
proh-relator-conduct\ (Ax\ref{ax:proh-relator-conduct})  & Thm/Uns & disability-block\ (Ax\ref{ax:disability-block})  & Uns     \\
proh-relator-remedy\ (Ax\ref{ax:proh-relator-remedy})   & Thm     & odrl-rel-is-rel\ \ (Ax\ref{ax:odrl-rel-is-rel})& Thm     \\
obl-relator\ \ \ \ \ (Ax\ref{ax:obl-relator})   & Thm/Sat & A1--A3\ \ & Thm \\
unique-founding ×6\ (Ax\ref{ax:unique-founding})      & Thm     & consistency\ \ \ \ \ (---) & Sat     \\
odrl-rel-typing ×3\  (Ax\ref{ax:odrl-rel-typing})    & Thm     &                            &         \\
\bottomrule
\end{tabular}
\end{table}
\vspace{-15pt}
\section{Related Work}
\label{sec:related}
A formal semantics for ODRL is still emerging. The behaviour of an evaluator is described in the \emph{ODRL Formal Semantics}, a W3C ODRL Community Group draft report (10 March 2026)~\cite{odrl-formal-semantics}. The draft specifies rule applicability and the evaluation of obligations and prohibitions, providing a foundation for formal reasoning about ODRL. However, aspects such as remedy evaluation and violation handling remain underspecified and are expressed at the level of evaluator decisions rather than a formal account of authority and enforcement~\cite{odrl-formal-semantics}. Steyskal and Polleres~\cite{steyskal2015odrl}, and Bonatti et al.~\cite{bonatti2025odrl}, provide formal and trace-based semantics for ODRL~2.1 and~2.2, capturing sanctioned interpretations and open/closed-world assumptions. Salas et al.~\cite{salas2026normalisation,salas2025odrl} further develop semantic frameworks for policy comparison, query answering, and normalisation. However, these approaches do not explicitly represent the normative positions underlying evaluator decisions or competence-level authority. A survey of usage control frameworks identifies the lack of formal semantics as a persistent gap~\cite{akaichi2022usagecontrolspecificationenforcement}. Dam et al.~\cite{dam2023policy} instantiate this issue through 22 policy patterns covering achievement, maintenance, and process obligations, but without an explicit account of the normative positions or authority roles they induce. Rodríguez-Doncel and Roman~\cite{rodriguez2025conformance} further show that conformance applies to policies but not processors, highlighting the lack of a unified theoretical grounding for policy enforcement. GUCON~\cite{akaichi2023gucon} provides a formal policy meta-model based on graph patterns and deontic logic, which can be instantiated using OWL~2, SHACL, and ODRL.

\noindent We adopt UFO-L as the foundational ontology due to its representation of legal and social normative concepts such as commitments, claims, and powers, which align closely with deontic notions in policy languages. In contrast, foundational ontologies such as BFO~\cite{Otte2022BFOBF} and DOLCE~\cite{gangemi2002sweetening} provide general ontological structures but do not explicitly model normative and institutional aspects required for representing policy authority, obligations, and permissions. De~Vos et al.~\cite{devos2019odrl} translate ODRL into Answer Set
Programming using deontic fluents for permission, prohibition, and
institutional power, but do not formalise the underlying normative
positions between parties, as is also the case for input/output
logic~\cite{makinson2000inputoutput} and defeasible deontic
logic~\cite{Governatori2013-GOVCSA}. Oliveira et al.~\cite{oliveira2024legion} align ODRL with UFO-L legal moments via OWL class equivalences, but without a full axiomatic grounding of correlative normative and competence-level structures. The Soberana Ontology~\cite{cunha2025soberana} applies UFO to model business ecosystems in data spaces, whereas our work focuses on the deontic layer for policy specification.
\vspace{-15pt}
\section{Conclusion}
\label{sec:conclusion}
An ODRL evaluator returns a verdict, but not what normative situation
a policy creates, who holds the authority behind it, or who may
declare it violated; a formal semantics fixes the verdict, not what
exists once a policy is in force. Answering that needs an ontology of
legal positions rather than another account of verdicts, which is why
we grounded ODRL in the foundational ontology UFO-L rather than giving
it a denotational semantics.
Our central result is the Cross-Level Design Principle: any policy
language whose norms can be broken and carry consequences needs
conduct-level positions (who may or must do what) together with
competence-level positions (who may declare a violation and impose its
consequences), because a conduct-only vocabulary cannot ground the
step from violation to remedy. This is a necessity result, not a
richer mapping, and it is what separates the work from earlier
conduct-level-only alignments. Applying it to ODRL, we map each
activated rule to a concrete legal relationship between assignee and
assigner, surface the correlatives an evaluator leaves silent (the
assigner's entitlement and the authority to declare a violation),
widen what a policy makes explicit from two positions to eight, and
let a companion profile state that authority at authoring time; the
whole theory is checked mechanically.

Those checks rely on general-purpose provers because no ODRL engine
yet covers the full vocabulary and constraint language, and the same
gap makes the grounding reusable: the profile lets standard reasoners
carry it on an evaluator's verdicts, querying authority, detecting
conflicts, and grounding remedies without a bespoke engine. In the
opening scenario the two providers' authority structures come out
distinct, so which provider may demand which remedy is fixed by the
model, not the implementation. The grounding covers simple legal
relationships only; richer constructs such as the Liberty relator lie
outside its scope.
\noindent\textbf{Future Work.}
(1)~extend conflict detection to inheritance-based
(\texttt{odrl:inheritFrom}) policies; and (2)~connect the grounding to
LLM-based ODRL generation, where the surfaced positions give an
ontological target for policy extraction and compliance checking.

\noindent
\\
\textbf{Declaration on Generative AI:}
The authors used Claude (Anthropic) and Grammarly for grammar and sentence refinement,
reviewed all content, and take full responsibility for the publication's content.

\noindent
\\
\textbf{Acknowledgments:} Supported by a Fraunhofer ICON grant through the Next Generation Dataspaces Initiative (NGDI), and the German Ministry for Research and Education (BMBF) project WestAI (Grant no. 01IS22094D). Also supported by the JUPITER AI Factory (JAIF) Project, funded via the EuroHPC Joint Undertaking (JU) under Grant ID 101250682 (\url{https://cordis.europa.eu/project/id/101250682}); and in part by the German Federal Government and the German States of North Rhine-Westphalia and Hesse, coordinated by the Forschungszentrum Jülich and the involving consortium partners, including RWTH Aachen University, Fraunhofer, and the Hessian AI Ecosystem.
\newpage
\bibliographystyle{vancouver}
\bibliography{references}
\end{document}